\pdfoutput=1
\documentclass[10pt,a4paper]{article}

\usepackage[a4paper,margin=2cm]{geometry}
\usepackage[T1]{fontenc}
\usepackage[utf8]{inputenc}
\usepackage{microtype}

\usepackage{amsmath,amssymb,amsthm}
\usepackage{aliascnt}
\usepackage{booktabs}
\usepackage{tikz}
\usetikzlibrary{positioning, calc, arrows.meta, decorations.pathreplacing}
\usepackage[noend]{algpseudocode}
\usepackage{algorithm}
\usepackage{float}
\usepackage{enumitem}

\usepackage[hidelinks]{hyperref}
\hypersetup{
  pdftitle={A 5/4 Bound for Graphic s-t Path TSP on Subcubic Graphs},
  pdfauthor={Junho Hwang},
  pdfkeywords={graphic TSP, path TSP, subcubic graphs, Held-Karp
    relaxation, even covers},
}
\usepackage[capitalize,noabbrev]{cleveref}

\theoremstyle{plain}
\newtheorem{theorem}{Theorem}
\newaliascnt{lemma}{theorem}
\newtheorem{lemma}[lemma]{Lemma}
\aliascntresetthe{lemma}
\newaliascnt{corollary}{theorem}
\newtheorem{corollary}[corollary]{Corollary}
\aliascntresetthe{corollary}
\newaliascnt{proposition}{theorem}

\aliascntresetthe{proposition}
\theoremstyle{definition}
\newaliascnt{definition}{theorem}

\aliascntresetthe{definition}
\newaliascnt{remark}{theorem}

\aliascntresetthe{remark}
\newtheorem*{remarkx}{Remark}

\crefname{theorem}{Theorem}{Theorems}
\Crefname{theorem}{Theorem}{Theorems}
\crefname{lemma}{Lemma}{Lemmas}
\Crefname{lemma}{Lemma}{Lemmas}
\crefname{corollary}{Corollary}{Corollaries}
\Crefname{corollary}{Corollary}{Corollaries}
\crefname{proposition}{Proposition}{Propositions}
\Crefname{proposition}{Proposition}{Propositions}
\crefname{definition}{Definition}{Definitions}
\Crefname{definition}{Definition}{Definitions}
\crefname{remark}{Remark}{Remarks}
\Crefname{remark}{Remark}{Remarks}

\crefname{algorithm}{Algorithm}{Algorithms}
\Crefname{algorithm}{Algorithm}{Algorithms}

\newcommand{\OPT}{\mathrm{OPT}}
\newcommand{\exc}{\mathrm{exc}}
\newcommand{\excR}[2]{\mathrm{exc}_{#2}\!\left(#1\right)}
\newcommand{\TSPtour}{\mathrm{tsp}}
\newcommand{\LP}{\mathrm{LP}}
\newcommand{\dfc}{\mathrm{def}}            
\newcommand{\dfce}{\widehat{\mathrm{def}}} 

\title{A \texorpdfstring{$5/4$}{5/4} Bound for Graphic
  \texorpdfstring{$s$--$t$}{s-t} Path TSP on Subcubic Graphs}
\author{Junho Hwang\thanks{Corresponding author}\\[2pt]
\normalsize Independent Researcher, Seoul, South Korea\\
\normalsize\texttt{junhohwang98@gmail.com}}
\date{}

\begin{document}

\maketitle

\begin{abstract}
We study the graphic \(s\)-\(t\) path TSP on subcubic graphs (maximum
degree~\(3\)): given distinct vertices \(s,t\), find a shortest
\(s\)-\(t\) walk that visits every vertex.  We prove an upper bound
with the asymptotically optimal leading coefficient \(5/4\) for every
terminal pair, even when \(G-\{s,t\}\) is disconnected.  Specifically,
every simple \(2\)-connected subcubic graph \(G\) on \(n\) vertices has
a spanning \(s\)-\(t\) walk of length at most
\(\lfloor(5n+n_2(G))/4\rfloor\), where \(n_2(G)\) counts its
degree-\(2\) vertices.  An \(O(n^2)\)-time algorithm attains this
bound.

Combining an edge-rooted even-cover theorem of Wigal, Yoo, and Yu
(WYY) with an even-cover-to-walk lemma proved here yields this bound
for adjacent terminals, a consequence not stated explicitly in their
paper.  We extend the
bound to arbitrary terminal pairs.  For cubic graphs, it becomes
\(\lfloor5n/4\rfloor\), to our knowledge the first direct \(5/4\)
bound for cubic path TSP that does not use the general path-to-tour
reduction.
\end{abstract}


\section{Introduction}\label{sec:intro}

In the graphic \(s\)-\(t\) \emph{path TSP},\footnote{Solutions are
spanning \emph{walks} and may repeat vertices or edges.  We retain the
standard name \emph{path TSP} from the metric literature
\cite{AKS15,SV14,TVZ20}; Wigal, Yoo, and Yu~\cite{WYY22} use ``TSP
walks'' for the tour problem.} the input is a connected unweighted
graph~\(G\) on \(n\) vertices and distinct terminals \(s,t\).  The
objective is a shortest \(s\)-\(t\) walk that visits every vertex.
Equivalently, \(\OPT(G,s,t)\) is the minimum number of edges in a
connected spanning multigraph whose edges are copies of edges of \(G\)
and whose odd-degree vertex set is \(\{s,t\}\).  We study subcubic
graphs, that is, graphs of maximum degree~\(3\).

\paragraph*{The Wigal--Yoo--Yu Theorem.}
Our main tool is a theorem of Wigal, Yoo, and Yu
(WYY)~\cite{WYY22}, which we state in elementary terms.  An
\emph{even cover} of \(G\) is a
spanning subgraph in which every vertex has degree \(0\) or \(2\); so
it is a disjoint union of cycles together with some isolated vertices.
Its \emph{excess} is
\[
  \exc(F)\;=\;2\cdot(\#\text{cycles of }F)\;+\;(\#\text{isolated vertices of }F).
\]
Fix an edge \(e=uv\) and an even cover \(F\) containing \(e\).
\Cref{lem:walk} gives a spanning \(u\)-\(v\) walk with at most
\(n+\exc(F)-3\) edges: delete \(e\), contract every component of
\(F-e\), and add two copies of the preimages of the edges in a spanning
tree of the contracted graph.  WYY guarantee an even cover whose
excess makes this length at most
\(\lfloor(5n+n_2(G))/4\rfloor-1\).  The walk must, however, join the
endpoints of the root edge.  We state this edge-endpoint consequence
as \cref{lem:rooted-walk}; it combines WYY's excess bound with
\cref{lem:walk} and is not stated in~\cite{WYY22}.

\paragraph*{Our Result.}
In path TSP, the prescribed terminals \(s,t\) need not be adjacent.
Our main theorem extends this bound to every terminal pair, losing
at most one edge in a single parity case.

\begin{theorem}[Subcubic endpoint conversion]\label{thm:endpoint-intro}
Let \(G\) be a simple \(2\)-connected subcubic graph on \(n\ge 3\)
vertices and let \(s,t\in V(G)\) be distinct.  Then
\[
  \OPT(G,s,t)\;\le\;\Bigl\lfloor\tfrac{5n+n_2(G)}{4}\Bigr\rfloor-1 ,
\]
except possibly when \(st\notin E(G)\), \(\deg_G(s)=\deg_G(t)=3\), and
\(n+n_2(G)\equiv2\pmod4\), in which case
\(\OPT(G,s,t)\le\lfloor(5n+n_2(G))/4\rfloor\).  The bound sharpens to
\(\lfloor(5n+n_2(G))/4\rfloor-2\) when \(st\notin E(G)\) and both
terminals have degree~\(2\).
\end{theorem}

For cubic graphs, the theorem gives
\(\OPT(G,s,t)\le\lfloor5n/4\rfloor-1\) whenever \(st\in E(G)\) or
\(n\equiv0\pmod4\), and \(\lfloor5n/4\rfloor\) otherwise.  A natural approach would
combine the WYY tour bound with the path-to-tour reduction of
Traub--Vygen--Zenklusen~\cite{TVZ20}.  For every fixed
\(\varepsilon>0\), their reduction converts an \(\alpha\)-approximation
for tours into a polynomial-time \((\alpha+\varepsilon)\)-approximation
for paths.  It calls the tour algorithm on subgraphs, so a
class-specific guarantee transfers only to classes closed under
subgraphs~\cite{TVZ20}; simple \(2\)-connected subcubic graphs are not
such a class.  Thus WYY alone does not yield even a
\((5/4+\varepsilon)\)-approximation by this route without an additional
decomposition argument.  Moreover, the reduction transfers ratios,
not explicit walk-length bounds of the form \((5n+n_2(G))/4\).

We instead prove the \(5/4\) coefficient directly.  The main task is
to convert arbitrary terminals into the endpoints of a root edge
while preserving the coefficient.  When both terminals have
degree~\(3\), adding \(st\) would violate subcubicity.

The theorem retains the full \((5n+n_2(G))/4\) expression on the stated
subcubic class and applies even when \(G-\{s,t\}\) is disconnected
(\cref{lem:component-aware}).  The algorithm in \cref{cor:algo-wyy}
runs in \(O(n^2)\) time.  For context, van Zuylen~\cite{vanZuylen18}
gives a \(5/4\)-approximation for cubic bipartite tours; our theorem
concerns fixed-endpoint walks on simple \(2\)-connected subcubic
graphs.

\paragraph*{Related Work.}
For metric path TSP, An, Kleinberg, and Shmoys~\cite{AKS15} gave the
first approximation ratio below \(5/3\), and
Traub--Vygen--Zenklusen~\cite{TVZ20} later transferred the tour ratio
to paths up to an additive~\(\varepsilon\).  In the graphic setting,
the technique of M\"omke and Svensson~\cite{MS16} underlies the
\(7/5\)-approximation for TSP and the \(3/2\)-approximation for path TSP
of Seb\H{o} and Vygen~\cite{SV14}.  Traub and Vygen~\cite{TV19} improve
the integrality ratio for graphic \(s\)-\(t\) tours.  These results have
no degree bound; under maximum degree~\(3\), we obtain the explicit
walk-length bound in \cref{thm:endpoint-intro}.

\paragraph*{Proof Idea.}
The reduction depends on the terminal degrees.  In each case,
\cref{lem:accounting} bounds the walk length by
\(\tfrac{5n+n_2(G)}{4}+\rho-1\), so it suffices to bound \(\rho\).
If both terminals have degree~\(2\), adding \(st\) gives \(\rho=-1\)
and the stronger additive term \(-2\).  If exactly one has degree~\(2\),
splitting the other terminal and subdividing the split edge gives
\(\rho=0\).  Neither case requires \(G-\{s,t\}\) to be connected.

If both terminals have degree~\(3\), adding \(st\) would create
degree-$4$ vertices.  Instead, we replace each terminal by two adjacent
vertices, add a root edge \(e^\star\) between the replacements, and
subdivide both split edges.  We apply WYY at \(e^\star\) and contract
back to \(G\).  The auxiliary instance raises the length bound by
\(4+\tfrac12\) edges and the contraction deletes at least four
edge-copies, hence \(\rho=\tfrac12\), which costs one edge only when
\((5n+n_2(G))/4\) is a half-integer.  A component-aware assignment
handles disconnected \(G-\{s,t\}\) (\cref{lem:component-aware}).

The coefficient \(5/4\) is best possible already on cubic graphs.

\begin{theorem}[Tightness of the $5/4$ coefficient]\label{thm:tight-intro}
There are infinitely many simple \(2\)-vertex-connected cubic graphs \(G\)
with an edge \(st\in E(G)\) such that
\[
  \OPT(G,s,t) \ge 5|V(G)|/4 - O(1).
\]
In particular, no bound of the form \(\OPT(G,s,t)\le \alpha n+O(1)\)
with \(\alpha<5/4\) can hold under the hypotheses of
\cref{thm:endpoint-intro}, even in the cubic special case.
\end{theorem}

\paragraph*{An Integrality-Gap Window.}
By \cref{lem:lp-identity}, the path Held--Karp LP value satisfies
\(\LP(G,s,t)\ge n-1\).  Since the LP is a relaxation,
\(\OPT(G,s,t)\ge\LP(G,s,t)\).  On
cubic graphs, \cref{thm:endpoint-intro} already bounds the path
Held--Karp gap by \(5/4+O(1/n)\).  On simple cubic \(3\)-edge-connected
graphs we can confine the worst-case asymptotic gap
\(\Gamma_{\mathrm{cub}}\) (defined in \cref{sec:gap-lower}) to a
window.

\begin{theorem}[Two-sided asymptotic gap window]\label{thm:gap-intro}
\[
  \frac98 \;\le\; \Gamma_{\mathrm{cub}} \;\le\; \frac54 .
\]
\end{theorem}

\noindent
The upper bound is \cref{thm:endpoint-intro}; the lower bound combines
the cubic tour examples of Luko\v{t}ka--Maz\'ak~\cite{LM18} with an
exact LP identity \(\LP(G,s,t)=n-1\) at adjacent terminals
(\cref{lem:adjacent-lp-exact}).  The exact value of
\(\Gamma_{\mathrm{cub}}\) is open.

\section{Preliminaries}\label{sec:prelim}

Throughout, ``$2$-connected'' means $2$-vertex-connected, and
edge-connectivity is always stated explicitly, as in
``$3$-edge-connected''.  All graphs are simple
unless stated otherwise.  For a graph $G$, $n(G):=|V(G)|$, $n_2(G)$
denotes the number of degree-$2$ vertices of $G$, and $d_G(u,v)$ is
the number of edges of a shortest $u$-$v$ path in $G$.  For the input
graph we abbreviate $n:=n(G)$; all other graphs carry explicit
arguments, as in $n(H)$, $n_2(K)$.

Let $G=(V,E)$ be a connected unweighted graph and let $s,t\in V$ be
distinct.
For distinct $u,v\in V$, a \emph{$\{u,v\}$-join} is a multigraph whose
edges are copies of edges of $G$ and whose set of odd-degree vertices
is exactly $\{u,v\}$; a connected spanning $\{u,v\}$-join admits an
Euler trail from $u$ to $v$.  We write $\OPT(G,s,t)$ for the graphic
$s$-$t$ path TSP optimum.  This is the minimum size of a connected
spanning $\{s,t\}$-join of $G$; equivalently, it is the length of a
shortest spanning $s$-$t$ walk in $G$.

For the lower bounds we use the \emph{path Held--Karp LP} on the
graphic metric of $G$.  Its variables are indexed by
$\binom{V}{2}$, and the cost of $uv$ is
$c_{uv}:=d_G(u,v)$.  Thus $c_e=1$ for $e\in E(G)$ and $c_e\ge2$
otherwise.  For $S\subseteq V$, let $\delta(S)$ be the corresponding
cut of the complete graph (and $\delta(v):=\delta(\{v\})$), and set
$x(F):=\sum_{e\in F}x_e$.  The LP is
\begin{align*}
  \LP(G,s,t) = \min\ & \textstyle\sum_{e\in\binom{V}{2}} c_e \, x_e \\
  \text{s.t.}\
    & x(\delta(v)) = 1\ \text{for}\ v \in \{s, t\},\quad
      x(\delta(v)) = 2\ \text{for}\ v \in V \setminus \{s, t\},\\
    & x(\delta(S)) \ge 1\ \text{for}\ \emptyset\ne S \subsetneq V,\;
      |S \cap \{s, t\}| = 1, \\
    & x(\delta(S)) \ge 2\ \text{for}\ \emptyset\ne S \subsetneq V,\;
      |S \cap \{s, t\}| \ne 1, \\
    & 0 \le x_e \le 1\ \text{for every}\ e\in\tbinom{V}{2}.
\end{align*}
Its integer optimum is $\OPT(G,s,t)$.  Following~\cite{SV14,TVZ20}, we include the box
constraints $0\le x_e\le 1$; our bounds do not depend on them.

\begin{lemma}\label{lem:lp-identity}
For every connected graph $G$ on $n$ vertices and distinct $s,t$, the path
Held--Karp LP on the graphic metric satisfies $\LP(G,s,t)\ge n-1$.
\end{lemma}

\begin{proof}
For any feasible $x$, summing the singleton degree constraints gives
$2\sum_e x_e=\sum_v x(\delta(v))=2(n-2)+2=2(n-1)$, so $\sum_e x_e=n-1$;
since $c_e\ge 1$ on the graphic metric, $\LP(G,s,t)\ge n-1$.
\end{proof}

\section{The Endpoint-Conversion Theorem}\label{sec:wyy}

This section proves \cref{thm:endpoint-intro,thm:tight-intro}.  We use
two results from WYY: the existential deficit bound with its equality
case (\cref{thm:wyy}), and the constructive \textsc{Scan} guarantee
(\cref{lem:wyy-constructive}).  Our \cref{lem:walk} converts an
edge-rooted even cover into a spanning walk; combining it with the
deficit bound gives \cref{lem:rooted-walk}.

\subsection{The Wigal--Yoo--Yu Excess Theorem}

Recall the even cover and its excess $\exc(F)=2c(F)+i(F)$ from the
introduction, where $c(F),i(F)$ count the cycles and isolated
vertices of $F$.  For an edge $e$ that lies in an even cover, let
$\mathcal{E}(G,e)$ be the set of even covers containing $e$ and put
$\excR{G}{e}:=\min_{F\in\mathcal{E}(G,e)}\exc(F)-2$.  The subtracted $2$
is the contribution of the cycle through~$e$ that is forced in every
$F\in\mathcal{E}(G,e)$, so $\excR{G}{e}$ measures excess \emph{beyond}
that unavoidable cycle.

WYY compare $\excR{G}{e}$ with the threshold
$(n(G)+n_2(G))/4$.  Define the \emph{rooted-excess deficit} by
\[
  \dfc(G,e):=\excR{G}{e}-\tfrac{n(G)+n_2(G)}{4}.
\]
Under the hypotheses of \cref{thm:wyy}, this quantity is a
half-integer and is at most $-1/2$.

\Cref{lem:walk} converts an even cover $F\in\mathcal{E}(G,e)$ into a
spanning walk of $G-e$ between the endpoints of $e$, of length at
most $n(G)+\exc(F)-3$.  Substituting
$\excR{G}{e}=\tfrac{n(G)+n_2(G)}{4}+\dfc(G,e)$ into this bound yields
the $5/4$ coefficient because
$n(G)+\tfrac{n(G)+n_2(G)}{4}=\tfrac{5n(G)+n_2(G)}{4}$.
In particular, decreasing the deficit by $1/2$ decreases the resulting
upper bound by $1/2$.  We state the combined bound in
\cref{lem:rooted-walk}.

We single out the equality case of \cref{thm:wyy}, which recurs in
our reductions (\cref{fig:theta-chain}).  A \emph{subcubic chain} from $u$ to $v$ is a
simple connected subcubic graph obtained by joining a sequence of
blocks by cut-edges.  Each block is either a $2$-connected subcubic
graph or a single vertex, and $u$ and $v$ are joined to the first and
last blocks by single \emph{chain edges}~\cite[Fig.~1]{WYY22}.  A \emph{rooted
$\theta$-chain} is a pair $(G,e)$, $e=uv$, such that $G-e$ is the
edge-disjoint union of two internally vertex-disjoint subcubic chains
from $u$ to $v$; together with $e$, they form three internally disjoint $u$-$v$
connections --- hence ``$\theta$''.

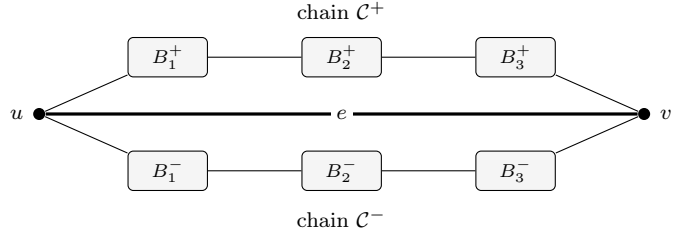
\begin{figure}[!htbp]
\centering
\begin{tikzpicture}[
  root/.style={circle, fill=black, inner sep=1.6pt},
  block/.style={draw, rounded corners=2pt, fill=black!4,
    minimum width=1.05cm, minimum height=0.52cm, inner sep=2pt,
    font=\scriptsize},
  lab/.style={font=\footnotesize},
  x=1cm, y=1cm,
]
\node[root, label={[lab]left:\(u\)}] (u) at (0,0) {};
\node[root, label={[lab]right:\(v\)}] (v) at (8,0) {};

\node[block] (a1) at (1.7,0.75) {\(B^+_1\)};
\node[block] (a2) at (4.0,0.75) {\(B^+_2\)};
\node[block] (a3) at (6.3,0.75) {\(B^+_3\)};
\draw (u) -- (a1) -- (a2) -- (a3) -- (v);

\node[block] (b1) at (1.7,-0.75) {\(B^-_1\)};
\node[block] (b2) at (4.0,-0.75) {\(B^-_2\)};
\node[block] (b3) at (6.3,-0.75) {\(B^-_3\)};
\draw (u) -- (b1) -- (b2) -- (b3) -- (v);

\draw[very thick] (u) -- (v) node[midway, fill=white, inner sep=2pt,
  font=\footnotesize] {\(e\)};
\node[lab, above] at (4.0,1.13) {chain \(\mathcal C^+\)};
\node[lab, below] at (4.0,-1.13) {chain \(\mathcal C^-\)};
\end{tikzpicture}
\caption{A rooted $\theta$-chain $(G,e)$, the equality case of
\cref{thm:wyy}.  Deleting the root edge $e=uv$ leaves two internally
vertex-disjoint subcubic chains $\mathcal C^+,\mathcal C^-$ from $u$
to $v$.  Each rounded box represents a block; nontrivial blocks are
$2$-connected, singleton blocks are also allowed, consecutive blocks
are joined by cut-edges, and the
edges at $u$ and $v$ are the chain edges.}
\label{fig:theta-chain}
\end{figure}

We distinguish two WYY bounds: the existential deficit $\dfc(G,e)$,
defined by minimising over all even covers, and the computable
estimate $\dfce(G,e)$ returned by WYY's linear-time \textsc{Scan}
procedure, which satisfies $\dfc(G,e)\le\dfce(G,e)\le-1/2$ and
controls the $O(n^2)$ construction.

In the following quoted result, we use the conventions of WYY: the
root edge may be a loop or have a parallel copy even though $G-e$ is
simple, and the one-vertex loop is admitted as a base instance.

\begin{theorem}[Wigal--Yoo--Yu~{\cite[Thm.~2.4~(T1)]{WYY22}}, in the form we use]\label{thm:wyy}
Let $G$ be a $2$-connected subcubic graph and $e=uv\in E(G)$ with
$G-e$ simple.  Then $\dfc(G,e)$ is a half-integer with
$\dfc(G,e)\le-1/2$, and equality holds only if $G$ is a single-vertex
loop or $(G,e)$ is a rooted $\theta$-chain.  In particular,
$\dfc(G,e)\le-1$ whenever $G$ is not a loop and $(G,e)$ is not a
rooted $\theta$-chain.
\end{theorem}

All our reductions use only \cref{thm:wyy}: Cases~A, C, and~D below
use the universal bound $\dfc\le-1/2$, and Case~B uses the refinement
$\dfc\le-1$, whose hypothesis \cref{lem:component-aware}
establishes.\footnote{\Cref{app:citations} tabulates the statements
of~\cite{WYY22} used in this paper, \cref{app:scan} records why the
\textsc{Scan} estimate satisfies the last sentence of
\cref{lem:wyy-constructive}, and \cref{app:impl} collects
implementation notes for \cref{alg:pathtsp}.}

\begin{remarkx}
Statement~(T3) of \cite[Thm.~2.4]{WYY22} also lists the instances
with $\dfc(G,e)=-1$, and an earlier version of this paper used it.
The list is incomplete: if $G_1$ arises from $K_4$ on
$\{u,v,z_1,z_2\}$ by subdividing $uz_1$ and $uz_2$, then every even
cover of $G_1$ through $uv$ is a $4$-cycle plus two isolated vertices
or a $5$-cycle plus one, so $\dfc(G_1,uv)=1-2=-1$, although
$G_1\not\cong K_4$, $uv$ has no parallel edge and lies in no
$2$-edge-cut, and $G_1-\{u,v\}$ is connected.  We rely only on~(T1),
whose proof does not use~(T3).
\end{remarkx}

\begin{lemma}[{\cite[Cor.~6.8]{WYY22}}]\label{lem:wyy-constructive}
Let $(G,e)$ be as in \cref{thm:wyy}.  In $O(n(G))$ time, WYY's
$\textup{\textsc{Scan}}$ computes a half-integer estimate $\dfce(G,e)$ with
$\dfc(G,e)\le\dfce(G,e)\le-1/2$.  A companion WYY routine outputs, in
$O(n(G)^2)$ time, an even cover $F\in\mathcal{E}(G,e)$ whose excess
satisfies
\[
  \exc(F)\le (n(G)+n_2(G))/4+\dfce(G,e)+2 .
\]
If $G-e$ is $2$-connected and $(G,e)$ is not a rooted $\theta$-chain,
then $\dfce(G,e)\le -1$.
\end{lemma}

\subsection{From an Edge-Rooted Even Cover to a Spanning Walk}

\begin{lemma}\label{lem:walk}
Let $G$ be a simple subcubic graph, let $e=uv\in E(G)$, and suppose
that $G-e$ is connected.  For every $F\in\mathcal{E}(G,e)$, the graph
$G-e$ admits a
spanning $u$-$v$ walk of length at most
$n(G) + \exc(F) - 3$.  Consequently,
$\OPT(G - e, u, v) \le n(G) + \excR{G}{e} - 1$.
\end{lemma}

\begin{proof}
Let $c:=c(F)$ and $i:=i(F)$ be the numbers of cycles and isolated
vertices of $F$, so $|E(F)|=n(G)-i$.  Deleting $e$ turns the cycle of
$F$ containing~$e$ into a $u$-$v$ path.  Thus $F-e$ has $c+i$
components, $|E(F-e)|=n(G)-i-1$, and odd-degree set $\{u,v\}$.

Contract every component of $F-e$ in $G-e$ and delete the resulting
loops.  The resulting multigraph $Q$ is connected and has $c+i$
vertices.  Choose a spanning tree of $Q$ and, for each tree edge,
choose a preimage in $G-e$.  Let $T$ be the resulting set of $c+i-1$
edges.  Every edge of $T$ lies outside $F$.  Adding two copies of every
edge of $T$ to $F-e$ yields a connected spanning $\{u,v\}$-join $W$
with
\[
  |E(W)| = (n(G) - i - 1) + 2(c + i - 1) = n(G) + \exc(F) - 3 .
\]
An Euler trail of $W$ from $u$ to $v$ has this length.  Since
$\min_{F\in\mathcal{E}(G,e)}\exc(F)=\excR{G}{e}+2$, minimising over
$F$ gives the stated bound.
\end{proof}

The next lemma records the edge-endpoint walk bound obtained by
combining \cref{lem:walk,thm:wyy}.  WYY do not state this consequence:
their theorem concerns rooted excess rather than \(s\)-\(t\) walks.

\begin{lemma}[Rooted walk bound]\label{lem:rooted-walk}
Let $H$ be a simple $2$-connected subcubic graph and let
$e=uv\in E(H)$.  If $\dfc(H,e)\le d$, then $H-e$ admits a spanning
$u$-$v$ walk of length at most $(5n(H)+n_2(H))/4+d-1$.
\end{lemma}

\begin{proof}
Since $H$ is $2$-connected, $H-e$ is connected and simple, so
\cref{lem:walk} gives $\OPT(H-e,u,v)\le n(H)+\excR{H}{e}-1$; now
substitute the definition
$\excR{H}{e}=(n(H)+n_2(H))/4+\dfc(H,e)$ and the bound
$\dfc(H,e)\le d$.
\end{proof}

We first prove the existential statement in
\cref{thm:endpoint-intro}; the $O(n^2)$-time construction is given in
\cref{cor:algo-wyy}.  We distinguish cases by whether $s,t$ are
adjacent and by their degrees.  If $st\in E(G)$ (Case~A), the
inequality $\dfc(G,st)\le-1/2$ applies directly.  If $st\notin E(G)$,
there is no edge with endpoints $s,t$ to which \cref{lem:walk} can be
applied.  We therefore construct an auxiliary rooted instance; the
construction depends on the terminal degrees (Cases~B--D).  The
$\{3,3\}$ construction appears in \cref{sec:case-32}, and the
degree-$2$ cases are handled in \cref{sec:mainproof}.

\subsection{The Split-and-Subdivide Construction}\label{sec:case-32}

Only the $\{3,3\}$ case uses this construction.

\paragraph*{Construction of $J^{33}$ and $K^{33}$
($\deg_G(s)=\deg_G(t)=3$).}
Write $N_G(s)=\{a,b,c\}$ and $N_G(t)=\{a',b',c'\}$.  Because $s$ and
$t$ are non-adjacent degree-$3$ vertices, the six incident edges
$sa,sb,sc,ta',tb',tc'$ are distinct, although neighbour labels may
coincide across the two terminals.  Replace $s$ by adjacent vertices
$s_1,s_2$, reattaching $sa,sb$ to $s_1$ and $sc$ to $s_2$.  Define
$t_1,t_2$ analogously, reattaching $ta',tb'$ to $t_1$ and $tc'$ to
$t_2$.  Add the root edge $e^\star:=s_2t_2$ to obtain $J^{33}$.
Finally, subdivide $s_1s_2$ at a new vertex $p$ and $t_1t_2$ at a new
vertex $q$ to obtain $K^{33}$ (see \cref{fig:split-gadget}).

\begin{figure}[!htbp]
\centering
\begin{tikzpicture}[
  vtx/.style={circle, draw, fill=black, inner sep=1.25pt},
  lbl/.style={font=\scriptsize},
  lab/.style={font=\footnotesize},
  every edge/.style={draw, line width=0.6pt},
  x=0.95cm, y=0.95cm,
]

\node[lab] at (1.0,1.4) {(1) original \(G\), \(st\notin E\)};
\node[vtx, label={[lbl]below:\(s\)}] (g-s) at (0.15,0) {};
\node[vtx, label={[lbl]below:\(t\)}] (g-t) at (1.85,0) {};
\draw (g-s) -- ++(-0.7, 0.55) node[lbl,left] {\(a\)};
\draw (g-s) -- ++(-0.7, 0.00) node[lbl,left] {\(b\)};
\draw (g-s) -- ++(-0.7,-0.55) node[lbl,left] {\(c\)};
\draw (g-t) -- ++( 0.7, 0.55) node[lbl,right] {\(a'\)};
\draw (g-t) -- ++( 0.7, 0.00) node[lbl,right] {\(b'\)};
\draw (g-t) -- ++( 0.7,-0.55) node[lbl,right] {\(c'\)};

\node[lab] at (7.0,1.4) {(2) split graph \(J^{33}\)};
\node[vtx, label={[lbl]above:\(s_1\)}] (j-s1) at (6.0,0.45) {};
\node[vtx, label={[lbl]below:\(s_2\)}] (j-s2) at (6.0,-0.45) {};
\node[vtx, label={[lbl]above:\(t_1\)}] (j-t1) at (8.0,0.45) {};
\node[vtx, label={[lbl]below:\(t_2\)}] (j-t2) at (8.0,-0.45) {};
\draw (j-s1) -- (j-s2);
\draw (j-t1) -- (j-t2);
\draw[very thick] (j-s2) -- (j-t2) node[midway, below=2pt, lbl] {\(e^\star\)};
\draw (j-s1) -- ++(-0.85, 0.40) node[lbl,left] {\(a\)};
\draw (j-s1) -- ++(-0.85, 0.05) node[lbl,left] {\(b\)};
\draw (j-s2) -- ++(-0.85,-0.05) node[lbl,left] {\(c\)};
\draw (j-t1) -- ++( 0.85, 0.40) node[lbl,right] {\(a'\)};
\draw (j-t1) -- ++( 0.85, 0.05) node[lbl,right] {\(b'\)};
\draw (j-t2) -- ++( 0.85,-0.05) node[lbl,right] {\(c'\)};

\node[lab] at (7.0,-2.6) {(3) subdivided graph \(K^{33}\)};
\node[vtx, label={[lbl]above=4pt:\(s_1\)}] (k-s1) at (6.0,-3.40) {};
\node[vtx, label={[lbl]left=2pt:\(p\)}]  (k-p)  at (6.0,-4.35) {};
\node[vtx, label={[lbl]below=4pt:\(s_2\)}] (k-s2) at (6.0,-5.30) {};
\node[vtx, label={[lbl]above=4pt:\(t_1\)}] (k-t1) at (8.0,-3.40) {};
\node[vtx, label={[lbl]right=2pt:\(q\)}] (k-q)  at (8.0,-4.35) {};
\node[vtx, label={[lbl]below=4pt:\(t_2\)}] (k-t2) at (8.0,-5.30) {};
\draw (k-s1) -- (k-p) -- (k-s2);
\draw (k-t1) -- (k-q) -- (k-t2);
\draw[very thick] (k-s2) -- (k-t2) node[midway, below=2pt, lbl] {\(e^\star\)};
\draw[dashed, rounded corners] (5.55,-5.45) rectangle (6.45,-3.30);
\draw[dashed, rounded corners] (7.55,-5.45) rectangle (8.45,-3.30);
\node[lbl] at (5.30,-4.35) {\(B_s\)};
\node[lbl] at (8.70,-4.35) {\(B_t\)};
\draw (k-s1) -- ++(-0.85, 0.45) node[lbl,left] {\(a\)};
\draw (k-s1) -- ++(-0.85, 0.15) node[lbl,left] {\(b\)};
\draw (k-s2) -- ++(-0.85, 0.00) node[lbl,left] {\(c\)};
\draw (k-t1) -- ++( 0.85, 0.45) node[lbl,right] {\(a'\)};
\draw (k-t1) -- ++( 0.85, 0.15) node[lbl,right] {\(b'\)};
\draw (k-t2) -- ++( 0.85, 0.00) node[lbl,right] {\(c'\)};

\node[lab] at (1.0,-2.6) {(4) delete \(e^\star\), contract \(B_s,B_t\)};
\node[vtx, label={[lbl]below:\(s\)}] (c-s) at (0.15,-4.35) {};
\node[vtx, label={[lbl]below:\(t\)}] (c-t) at (1.85,-4.35) {};
\draw (c-s) -- ++(-0.7, 0.55) node[lbl,left] {\(a\)};
\draw (c-s) -- ++(-0.7, 0.00) node[lbl,left] {\(b\)};
\draw (c-s) -- ++(-0.7,-0.55) node[lbl,left] {\(c\)};
\draw (c-t) -- ++( 0.7, 0.55) node[lbl,right] {\(a'\)};
\draw (c-t) -- ++( 0.7, 0.00) node[lbl,right] {\(b'\)};
\draw (c-t) -- ++( 0.7,-0.55) node[lbl,right] {\(c'\)};
\node[lab, align=center] at (1.0,-5.40) {at least four edge-copies become\\
loops and are deleted under the contraction};

\draw[-{Latex[length=2.5mm]}, thick] (3.0,0) -- (4.6,0)
  node[midway,above,lbl] {split, add \(e^\star\)};
\draw[-{Latex[length=2.5mm]}, thick] (7.0,-1.0) -- (7.0,-2.2)
  node[midway,right=2pt,lbl] {subdivide};
\draw[-{Latex[length=2.5mm]}, thick] (4.6,-4.35) -- (3.0,-4.35)
  node[midway,above,lbl] {delete \(e^\star\), contract};
\end{tikzpicture}
\caption{Endpoint conversion when \(st\notin E(G)\).  Splitting \(s\)
and \(t\) gives the auxiliary graph \(J^{33}\); subdividing the two
split edges gives \(K^{33}\); contracting \(B_s,B_t\) in
\(K^{33}-e^\star\) and deleting loops recovers \(G\).  Neighbour labels denote
incidences, and the same vertex may occur on both sides.}
\label{fig:split-gadget}
\end{figure}
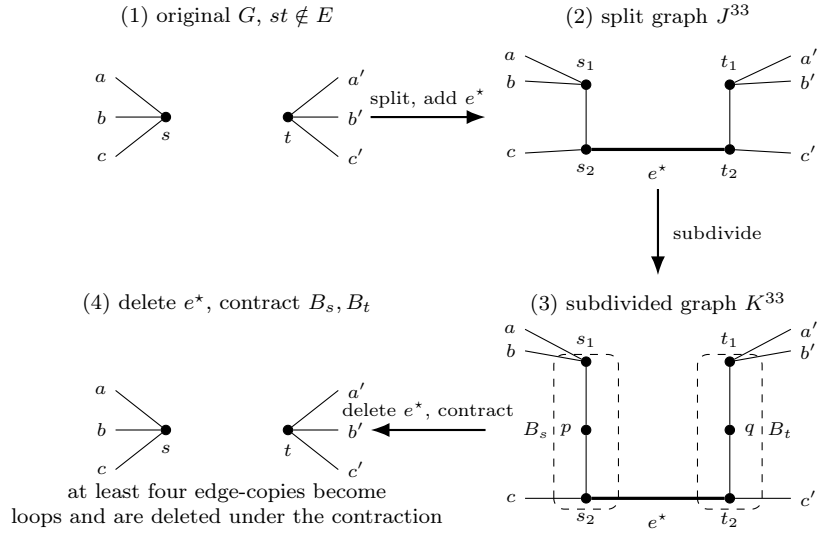

The contraction blocks are $B_s:=\{s_1,p,s_2\}$ and
$B_t:=\{t_1,q,t_2\}$.  Edges formerly incident with the same terminal
retain distinct other endpoints because $G$ is simple; edges formerly
incident with different terminals acquire distinct split endpoints.
Thus $K^{33}$ has no loops or parallel edges, even if some vertex is
adjacent to both $s$ and $t$ (\cref{lem:gadget-admissibility}(a)).

The root edge is $e^\star=s_2t_2$.  The subdivision vertices $p,q$
inside $B_s,B_t$ force at least four internal edge-copies in every
connected spanning $\{s_2,t_2\}$-join (\cref{lem:four-loop-loss});
this is the value $\lambda=4$ used in \cref{lem:accounting}.

\begin{lemma}[Splitting and subdivision]\label{lem:split-connectivity}
Let \(H\) be a \(2\)-connected graph, and let \(H'\) be
obtained from \(H\) by replacing a degree-\(3\) vertex \(v\) by two
adjacent vertices \(v_1,v_2\), reattaching two of the three edges at
\(v\) to \(v_1\) and the third to \(v_2\).  Then \(H'\) is
\(2\)-connected.  Likewise, subdividing an edge of a
\(2\)-connected graph yields a \(2\)-connected graph.
\end{lemma}

\begin{proof}
For the first claim it suffices to show that \(H'-x\) is connected
for every \(x\in V(H')\).  Suppose first that \(x\notin\{v_1,v_2\}\).
Contracting the edge \(v_1v_2\) in \(H'-x\) yields \(H-x\).
Contracting an edge preserves the number of connected components, so
\(H'-x\) is connected because \(H-x\) is.  Suppose next that
\(x=v_1\).  Then \(H'-v_1\) is obtained
from \(H-v\) by adding the vertex \(v_2\) and joining it to the one
neighbour of \(v\) whose edge was reattached to \(v_2\); since
\(H-v\) is connected, so is \(H'-v_1\).  Suppose finally that
\(x=v_2\).  Then \(H'-v_2\) is \(H\) with one edge at \(v\) deleted
(the vertex \(v_1\) playing the role of \(v\)); since \(H\) is
\(2\)-connected, it has no bridge, so \(H'-v_2\) is connected.

For the second claim, let \(H''\) be obtained from \(H\) by
subdividing an edge \(xy\) at a new vertex \(r\).  Deleting \(r\)
from \(H''\) leaves \(H-xy\), which is connected because \(H\) has no
bridge.  Now delete a vertex \(w\in V(H)\).  If \(w\notin\{x,y\}\),
then \(H''-w\) is \(H-w\) with the edge \(xy\) subdivided; if
\(w\in\{x,y\}\), then \(H''-w\) is \(H-w\) with \(r\) added and
joined to the surviving endpoint of \(xy\).  In both cases \(H''-w\)
is connected because \(H-w\) is.
\end{proof}

\begin{lemma}[Admissibility of $K^{33}$]\label{lem:gadget-admissibility}
Let $G$ be a simple $2$-connected subcubic graph, let
$s,t\in V(G)$ be non-adjacent with
$\deg_G(s)=\deg_G(t)=3$, and let $K=K^{33}$ be the subdivided graph of
\cref{sec:case-32}.  Then:
\begin{enumerate}[leftmargin=2em, label=\textup{(\alph*)}]
  \item $K$ is simple subcubic with $|V(K)|=n+4$ and
    $n_2(K)=n_2(G)+2$;
  \item $K-e^\star$ is $2$-connected;
  \item if additionally $G-\{s,t\}$ is connected, then $(K,e^\star)$
    is not a rooted $\theta$-chain.
\end{enumerate}
\end{lemma}

\begin{proof}
For part~(a), the preceding simplicity argument shows that $K$ has no
loops or parallel edges.  Replacing $\{s,t\}$ by
$\{s_1,s_2,t_1,t_2,p,q\}$ adds four vertices.  The split vertices have
degree~$3$, and $p,q$ are the only new degree-$2$ vertices.  Hence
$|V(K)|=n+4$ and $n_2(K)=n_2(G)+2$.

For part~(b), $K-e^\star$ is obtained from $G$ by splitting $s$ and $t$
and then subdividing the two split edges, so it is
$2$-connected by repeated application of
\cref{lem:split-connectivity}.

For part~(c), suppose that $G-\{s,t\}$ is connected.  After deleting
$s_2,t_2$, the vertices $p,q$ are leaves adjacent to $s_1,t_1$,
respectively, and each of $s_1,t_1$ has neighbours in the connected
graph $G-\{s,t\}$.  Thus $K-\{s_2,t_2\}$ is connected.  Suppose, for
contradiction, that $(K,e^\star)$ is a rooted $\theta$-chain.  Then
$K-e^\star$ is the edge-disjoint union of two internally
vertex-disjoint subcubic $s_2$-$t_2$ chains.  Neither chain can be a
    single edge, since $e^\star$ is the only edge between $s_2$ and $t_2$;
    hence
both chains have nonempty interiors.  No edge joins the two interiors,
so $K-\{s_2,t_2\}$ would be disconnected, a contradiction.
\end{proof}

\begin{lemma}[Contraction preserves the join]\label{lem:contraction}
Let $M$ be a connected spanning $\{u, v\}$-join multigraph, and let
$B_u, B_v \subseteq V(M)$ be disjoint vertex sets with $u \in B_u$,
$v \in B_v$, such that every vertex of
$(B_u \cup B_v) \setminus \{u, v\}$ has even $M$-degree.
Contracting $B_u$ to a single vertex $s$ and $B_v$ to a single
vertex $t$, and deleting the resulting loops, yields a connected
spanning $\{s, t\}$-join multigraph $M'$ on
$(V(M) \setminus (B_u \cup B_v)) \cup \{s, t\}$ with $|E(M')|$ equal
to $|E(M)|$ minus the number of internal edge-copies (edges with
both endpoints in $B_u$ or both in $B_v$) in $M$.
\end{lemma}

\begin{proof}
Each edge of $M$ with both endpoints in $B_u$ or both in $B_v$
becomes a loop after contraction and is deleted, which gives the
stated edge count.  Internal edges contribute twice to the sum of
degrees in $B_u$, and every vertex in $B_u\setminus\{u\}$ has even
$M$-degree.  Therefore
\[
  \deg_{M'}(s)
  \equiv \sum_{w\in B_u}\deg_M(w)
  \equiv \deg_M(u)
  \equiv 1 \pmod 2.
\]
The same argument applies to $t$, and all other vertices retain even
degree.  Contracting vertex sets of a connected multigraph and
deleting loops preserves connectivity.
\end{proof}

Each terminal conversion below produces a rooted instance
$(H,e^\star)$ and a pair of contraction blocks; a single accounting
lemma then gives the walk length.

\begin{lemma}[Reduction accounting]\label{lem:accounting}
Suppose a reduction replaces $(G,s,t)$ by a rooted instance $(H,e^\star)$,
$e^\star=uv$, meeting the hypotheses of \cref{lem:rooted-walk}, together
with disjoint contraction blocks $B_u\ni u$, $B_v\ni v$ as in
\cref{lem:contraction}.  Assume that the loopless quotient of
$H-e^\star$ obtained by contracting the blocks is a spanning
subgraph of $G$, with the contracted vertices identified with
$s,t$.  Assume also that every connected spanning $\{u,v\}$-join of
$H-e^\star$ contains at least $\lambda$ edge-copies internal to $B_u$
or $B_v$.  Put
\[
  \alpha:=n(H)-n,\qquad \beta:=n_2(H)-n_2(G),\qquad
  \rho:=\tfrac{5\alpha+\beta}{4}+d-\lambda ,
\]
where $\dfc(H,e^\star)\le d$.
Then $G$ has a spanning $s$-$t$ walk of length at most
$\tfrac{5n+n_2(G)}{4}+\rho-1$, hence
\[
  \OPT(G,s,t)\ \le\ \Bigl\lfloor\tfrac{5n+n_2(G)}{4}+\rho\Bigr\rfloor-1 .
\]
In particular $\OPT(G,s,t)\le\lfloor(5n+n_2(G))/4\rfloor-1$ whenever
$\rho\le 0$.
\end{lemma}

\begin{proof}
\Cref{lem:rooted-walk} gives a spanning $u$-$v$ walk of $H-e^\star$ of
length at most $\tfrac{5n(H)+n_2(H)}{4}+d-1
=\tfrac{5n+n_2(G)}{4}+\tfrac{5\alpha+\beta}{4}+d-1$.  Contracting
$B_u\mapsto s$, $B_v\mapsto t$ and deleting loops
(\cref{lem:contraction}) removes at least $\lambda$ edges and leaves a
connected spanning $\{s,t\}$-join of $G$ of length at most
$\tfrac{5n+n_2(G)}{4}+\rho-1$, realised as an $s$-$t$ walk by an Euler
trail.  Because $\OPT$ is integral and
$\lfloor x-1\rfloor=\lfloor x\rfloor-1$, this real-valued bound implies
the displayed inequality.  If $\rho\le0$, then
$\lfloor x+\rho\rfloor\le\lfloor x\rfloor$ for
$x=(5n+n_2(G))/4$, which proves the final assertion.
\end{proof}

\begin{lemma}[Contraction loss: at least four edge-copies]\label{lem:four-loop-loss}
Let $W$ be a connected spanning $\{s_2,t_2\}$-join multigraph of
$K-e^\star$ whose edges are copies of edges of $K-e^\star$.  After
contracting $B_s=\{s_1,p,s_2\}$ to $s$ and $B_t=\{t_1,q,t_2\}$ to $t$
and deleting loops, the resulting connected spanning
$\{s,t\}$-join multigraph in $G$ has at most $|E(W)|-4$ edges.
\end{lemma}

\begin{proof}
Since the odd-degree set of $W$ is $\{s_2,t_2\}$, both $p$ and $q$
have even $W$-degree.  Because $W$ is connected and spanning, both
degrees are positive and hence at least~$2$.  Every edge-copy incident with
$p$ is internal to $B_s$, and every edge-copy incident with $q$ is
internal to $B_t$.  Thus contraction deletes at least four
edge-copies.
\end{proof}

\begin{lemma}[Parity of $n(G)+n_2(G)$]\label{lem:parity-even}
For every $2$-connected subcubic graph $G$ with $n(G)\ge2$,
the integer $n(G)+n_2(G)$ is even.
\end{lemma}

\begin{proof}
If $n(G)=2$, the claim is immediate.  Assume henceforth that
$n(G)\ge3$.  A $2$-connected graph then has minimum degree at least~$2$.  Since $G$ is
subcubic, every vertex has degree $2$ or $3$; let $n_3(G)$ denote the
number of degree-$3$ vertices.  The handshaking lemma shows that
$2n_2(G)+3n_3(G)$ is even, so $n_3(G)$ is even.  Therefore
$n(G)+n_2(G)=n_3(G)+2n_2(G)$ is even.
\end{proof}

The construction still allows us to choose which two edges incident
with each terminal are assigned to its first split vertex.  We make
this choice using the components of $G-\{s,t\}$.  The resulting
\emph{component-aware split} works even when $G-\{s,t\}$ is
disconnected.

\begin{lemma}[Component-aware split handles the separating case]\label{lem:component-aware}
Let $G$ be a simple $2$-connected subcubic graph and let
$s,t\in V(G)$ be non-adjacent with $\deg_G(s)=\deg_G(t)=3$.  In the
$\{3,3\}$ construction of \cref{sec:case-32}, one can assign two
incidences to $s_1$ and one to $s_2$, and likewise at $t$, so that
$K^{33}-\{s_2,t_2\}$ is connected.  Consequently
$(K^{33},e^\star)$ is not a rooted $\theta$-chain, and
$\dfc(K^{33},e^\star)\le -1$.
\end{lemma}

\begin{proof}
Since $G$ is $2$-connected and $s,t$ are non-adjacent, every
component of $G-\{s,t\}$ has at least one edge to $s$ and at least one
edge to $t$; otherwise $s$ or $t$ would be a cut vertex.  As
$\deg_G(s)=3$, the number $k$ of components is at most $3$; write
$C_1,\dots,C_k$ for them.  The graph $K^{33}-\{s_2,t_2\}$ consists of
$G-\{s,t\}$, the vertex $s_1$ joined to its two assigned neighbours,
the vertex $t_1$ joined to its two assigned neighbours, and the
subdivision vertices $p,q$ pendant at $s_1,t_1$.  It therefore
suffices to choose the assignments so that $C_1,\dots,C_k$
together with $s_1,t_1$ form a connected graph.

If $k=1$, any assignment works.  If $k=2$, both $C_1$ and $C_2$ have
an edge to $s$.  Assign to $s_1$ one such edge from each component, so
that $s_1$ connects $C_1$ and $C_2$; choose any valid assignment at
$t$.  If $k=3$, each $C_i$ has exactly one edge to $s$ and one edge to
$t$.  Assign to $s_1$ the $s$-edges from $C_1,C_2$, and assign to
$t_1$ the $t$-edges from $C_1,C_3$.  Then $s_1$ connects $C_1,C_2$
and $t_1$ connects $C_1,C_3$.

In every case $K^{33}-\{s_2,t_2\}$ is connected, and the argument of
\cref{lem:gadget-admissibility}(c) --- which uses only this
connectedness --- shows that $(K^{33},e^\star)$ is not a rooted
$\theta$-chain, so \cref{thm:wyy}, applicable by
\cref{lem:gadget-admissibility}(a,b), gives $\dfc(K^{33},e^\star)\le-1$.
\end{proof}

\subsection{Proof of the Main Bound}\label{sec:mainproof}

We now assemble the endpoint reductions.  Each case supplies a rooted
instance and contraction blocks for \cref{lem:accounting}.

\begin{proof}[Proof of \cref{thm:endpoint-intro}]
We distinguish four cases according to adjacency and terminal degrees.

\smallskip
\noindent\emph{Adjacent terminals (Case~A).}
Take $(H,e^\star)=(G,st)$ with singleton blocks, so
$\alpha=\beta=\lambda=0$.  By the universal inequality in
\cref{thm:wyy}, we may take $d=-1/2$, and hence $\rho=-1/2$.

\smallskip
\noindent\emph{Non-adjacent degree-$3$ terminals (Case~B).}
Build $K^{33}$ with the component-aware split.  By
\cref{lem:component-aware}, $(K^{33},e^\star)$ is not a rooted
$\theta$-chain, so \cref{thm:wyy} allows $d=-1$.  The parameters are
$(\alpha,\beta,\lambda)=(4,2,4)$
(\cref{lem:gadget-admissibility}(a) and \cref{lem:four-loop-loss}),
hence $\rho=1/2$, and \cref{lem:accounting} gives
$\OPT(G,s,t)\le\lfloor x+\tfrac12\rfloor-1$ for $x:=(5n+n_2(G))/4$,
which is $\lfloor x\rfloor-1$ if $x$ is an integer and
$\lfloor x\rfloor$ if $x$ is a half-integer (\cref{lem:parity-even}).

\smallskip
\noindent\emph{Non-adjacent mixed degrees (Case~C).}
Relabel the terminals so that $\deg_G(s)=3$ and $\deg_G(t)=2$.  Write
$N_G(s)=\{a,b,c\}$.  Replace $s$ by adjacent vertices $s_1,s_2$,
reattaching $sa,sb$ to $s_1$ and $sc$ to $s_2$.  Add the root edge
$e^\star:=s_2t$ and subdivide $s_1s_2$ at a new vertex $p$.  Denote
the resulting graph by $K'$, and set $B_s:=\{s_1,p,s_2\}$ and
$B_t:=\{t\}$.  Then
$\deg_{K'}(s_1)=\deg_{K'}(s_2)=3$, $\deg_{K'}(p)=2$, and the degree
of $t$ increases from~$2$ to~$3$; all other degrees are unchanged.
The change at $t$ removes one degree-$2$ vertex, while $p$ adds one.
Hence $K'$ is simple subcubic with $n(K')=n+2$ and
$n_2(K')=n_2(G)$.  Applying \cref{lem:split-connectivity} first to the
split and then to the subdivision shows that $K'-e^\star$ is
$2$-connected and simple.
Thus $(K',e^\star)$ satisfies \cref{thm:wyy}, and we may take
$d=-1/2$.  In every connected spanning $\{s_2,t\}$-join of
$K'-e^\star$, the vertex $p$ has positive even degree.  Since every
edge incident with $p$ is internal to $B_s$, contraction deletes at
least two edge-copies.  Thus
$(\alpha,\beta,\lambda)=(2,0,2)$ and
$\rho=\tfrac52-\tfrac12-2=0$.  This case uses only the universal
deficit bound and therefore does not need the equality case of
\cref{thm:wyy} to be excluded.  In particular, it also applies when
$G-\{s,t\}$ is disconnected.

\smallskip
\noindent\emph{Non-adjacent degree-$2$ terminals (Case~D).}
Put $G':=G+st$ and $e^\star:=st$, with singleton blocks.  Then
$G'-e^\star=G$, and $G'$ is simple, subcubic, and
$2$-connected, so \cref{thm:wyy} applies and we may take $d=-1/2$.
Since $n(G')=n$ and $n_2(G')=n_2(G)-2$, we get
$(\alpha,\beta,\lambda)=(0,-2,0)$ and $\rho=-1$.

\smallskip
In Cases~A, C, and~D we have $\rho\le 0$, so \cref{lem:accounting}
yields $\OPT(G,s,t)\le\lfloor(5n+n_2(G))/4\rfloor-1$, sharpening to
$\lfloor(5n+n_2(G))/4\rfloor-2$ in Case~D, where $\rho=-1$; Case~B was
settled above.  The $O(n^2)$ construction is \cref{cor:algo-wyy}.
\end{proof}

\subsection{Algorithmic Form}\label{sec:algorithmic-form}

\begin{corollary}[Algorithmic form]\label{cor:algo-wyy}
Given a simple $2$-connected subcubic graph $G$ on $n\ge3$ vertices
and distinct terminals $s,t\in V(G)$, \cref{alg:pathtsp} outputs, in
$O(n^2)$ time, a spanning $s$-$t$ walk whose length satisfies the
bounds of \cref{thm:endpoint-intro}.
\end{corollary}

\begin{algorithm}[t]
\caption{Path TSP walk on simple $2$-connected subcubic graphs}
\label{alg:pathtsp}
\footnotesize
\begin{algorithmic}[1]
\Require simple $2$-connected subcubic $G$ on $n\ge 3$ vertices; distinct $s,t\in V(G)$
\Ensure spanning $s$-$t$ walk of $G$ satisfying the bounds of \cref{thm:endpoint-intro}
\If{$st\in E(G)$}\Comment{Case~A}
  \State $(H,e)\gets(G,st)$; $(u,v)\gets(s,t)$; $B_s\gets\{s\}$, $B_t\gets\{t\}$
\ElsIf{$\deg_G(s)=\deg_G(t)=2$}\Comment{Case~D}
  \State $(H,e)\gets(G+st,\,st)$; $(u,v)\gets(s,t)$; $B_s\gets\{s\}$, $B_t\gets\{t\}$
\ElsIf{$\{\deg_G(s),\deg_G(t)\}=\{3,2\}$}\Comment{Case~C; relabel so $\deg_G(s)=3$}
  \State split $s$ into $s_1,s_2$, add $e^\star=s_2t$, subdivide $s_1s_2$ at $p$, to get $K'$
  \State $(H,e)\gets(K',e^\star)$; $(u,v)\gets(s_2,t)$; $B_s\gets\{s_1,p,s_2\}$, $B_t\gets\{t\}$
\Else\Comment{Case~B, $\{3,3\}$}
  \State build $K^{33}$ with the component-aware split (\cref{lem:component-aware}), root edge $e^\star=s_2t_2$, subdivisions $p,q$
  \State $(H,e)\gets(K^{33},e^\star)$; $(u,v)\gets(s_2,t_2)$; $B_s\gets\{s_1,p,s_2\}$, $B_t\gets\{t_1,q,t_2\}$
\EndIf
\State $F\gets$ WYY even-cover routine on $(H,e)$\Comment{\cref{lem:wyy-constructive}}
\State $F'\gets F-e$; $M\gets$ \emph{component-contraction multigraph} of $(H-e)/F'$, one vertex per $F'$-component (\cref{lem:walk})
\State $T_M\gets$ spanning tree of the loopless part of $M$; lift each $T_M$-edge to a preimage in $E(H-e)\setminus E(F)$, forming $T_H$; $W\gets F'\cup 2T_H$
\State $W_G\gets W$ with $B_s\mapsto s$, $B_t\mapsto t$, loops deleted\Comment{\cref{lem:contraction}; singleton blocks in Cases~A,~D}
\State \Return Euler trail from $s$ to $t$ of $W_G$ via Hierholzer
\end{algorithmic}
\end{algorithm}

\begin{proof}
Apply the analysis in the proof of \cref{thm:endpoint-intro} with the
\textsc{Scan} bound $\dfce$ in place of $\dfc$
(\cref{lem:accounting}).  Cases~A, C, and~D use only the universal
estimate $\dfce\le-1/2$.  In Case~B, $K^{33}-e^\star$ is
$2$-connected (\cref{lem:gadget-admissibility}(b)) and
$(K^{33},e^\star)$ is not a rooted $\theta$-chain
(\cref{lem:component-aware}), so \cref{lem:wyy-constructive} gives
$\dfce(K^{33},e^\star)\le-1$, the value of $d$ used in the proof of
\cref{thm:endpoint-intro}.  All steps except the WYY procedure take
linear time, so the algorithm runs in $O(n^2)$.
\end{proof}

\subsection{Tightness of the \texorpdfstring{$5/4$}{5/4} Coefficient}\label{sec:tightness}

The $5/4$ coefficient in \cref{thm:endpoint-intro} is best possible
already in the cubic special case $n_2(G)=0$.  Let $\TSPtour(G)$ denote
the minimum length of a spanning closed walk in $G$.

\begin{proof}[Proof of \cref{thm:tight-intro}]
Dvo\v{r}\'ak, Kr\'al', and Mohar~\cite{DKM17} construct infinitely
many simple $2$-connected cubic graphs whose tour optimum satisfies
$\TSPtour(G)\ge 5n/4-2$ (see also \cite[Sec.~7]{WYY22}).  For any
edge $st\in E(G)$, adding one copy of $st$ to a spanning $s$-$t$ walk
yields a spanning closed walk.  Hence
$\OPT(G,s,t)\ge\TSPtour(G)-1\ge5n/4-3$.
These instances are cubic ($n_2(G)=0$) with adjacent terminals, so
they satisfy the hypotheses of \cref{thm:endpoint-intro}.  The claimed
asymptotic bound follows along this infinite family.
\end{proof}

\section{The Held--Karp Gap on Cubic \texorpdfstring{$3$}{3}-Edge-Connected Graphs}\label{sec:gap-lower}

In the cubic case \(n_2(G)=0\), \cref{thm:endpoint-intro} and the
elementary bound \(\LP\ge n-1\) give an upper bound on the path
Held--Karp gap.  A complementary lower-bound construction yields the
window in \cref{thm:gap-intro}.  Define
\[
  \Gamma_{\mathrm{cub}}
  := \limsup_{\substack{n\to\infty\\ n\text{ even}}}\ \,
     \sup_{\substack{G\text{ simple cubic }3\text{-edge-connected},\\
                      |V(G)|=n,\ s,t\in V(G),\ s\ne t}}
     \frac{\OPT(G,s,t)}{\LP(G,s,t)}
\]
to be the worst-case asymptotic gap over this graph class (every
cubic graph has even order).

\begin{corollary}[Gap upper bound]\label{cor:wyy-gap}
For every simple cubic \(3\)-edge-connected graph \(G\) on \(n\)
vertices and distinct \(s,t\in V(G)\),
\[
  \frac{\OPT(G,s,t)}{\LP(G,s,t)}\le\frac{\lfloor 5n/4\rfloor}{n-1}
  =\frac54+O(1/n).
\]
\end{corollary}

\begin{proof}
Divide the cubic bound \(\OPT(G,s,t)\le\lfloor 5n/4\rfloor\) of
\cref{thm:endpoint-intro} by \(\LP(G,s,t)\ge n-1\)
(\cref{lem:lp-identity}).
\end{proof}

In particular, on cubic graphs the algorithm of \cref{cor:algo-wyy}
is a \((5/4+O(1/n))\)-approximation for graphic \(s\)-\(t\) path TSP.

For the lower bound, we use the following fact: at adjacent terminals,
\(\LP=n-1\) on every simple cubic \(3\)-edge-connected graph.

\begin{lemma}[Adjacent-terminal LP certificate]\label{lem:adjacent-lp-exact}
Let \(G\) be a simple cubic \(3\)-edge-connected graph on \(n\)
vertices, and let \(st\in E(G)\).  Then \(\LP(G,s,t)=n-1\).
\end{lemma}

\begin{proof}
The lower bound \(\LP(G,s,t)\ge n-1\) is \cref{lem:lp-identity}.  For
the upper bound, choose a perfect matching \(M\) of \(G\) containing
\(st\); such an \(M\) exists because every edge of a bridgeless cubic
graph lies in a perfect matching~\cite{Petersen1891,Plesnik72,Schoenberger34}.
For $e\in\binom{V}{2}$, set $y_e=1$ if $e\in M$, set $y_e=1/2$ if
$e\in E(G)\setminus M$, and set $y_e=0$ if $e\notin E(G)$.  Then every
vertex has $y$-degree~$2$.

For a nonempty proper set $S\subsetneq V$, put
$\delta_G(S):=\delta(S)\cap E(G)$.  We claim that
$y(\delta(S))\ge2$.  This is immediate if $|\delta_G(S)|\ge4$.  If
$|\delta_G(S)|=3$, then $|S|$ is odd because $G$ is cubic.  Moreover,
since $M$ is perfect,
$k:=|M\cap\delta_G(S)|\equiv|S|\pmod2$.  Hence $k$ is odd and
$y(\delta(S))=(3+k)/2\ge2$.  Thus $y$ is feasible for the tour
Held--Karp relaxation and $y_{st}=1$.

Let $\chi^{st}$ be the indicator vector of $st$ and set
$x:=y-\chi^{st}$.  It decreases the degrees of $s,t$ and the value of
every $s$--$t$ cut by one, leaving all other cut values unchanged.
Therefore $x$ satisfies every path Held--Karp constraint.  Its cost is
$|M|+\tfrac12|E(G)\setminus M|-1=n-1$.
\end{proof}

Thus adjacent terminal pairs attain the minimum LP value, which
transfers asymptotic lower bounds for cubic tours to the path setting.

\begin{theorem}[\(\Gamma_{\mathrm{cub}}\ge 9/8\)]\label{thm:gap-lower}
The asymptotic path Held--Karp gap on simple cubic
\(3\)-edge-connected graphs is at least \(9/8\).
\end{theorem}

\begin{proof}
Luko\v{t}ka and Maz\'ak construct a sequence of simple
\(3\)-vertex-connected
cubic graphs \(G_k\), with \(n_k:=|V(G_k)|\to\infty\), whose graphic
tour optimum satisfies \(\TSPtour(G_k) \ge (9/8-o(1))n_k\)
\cite{LM18}.  Since \(3\)-vertex-connected cubic graphs are
\(3\)-edge-connected,
these graphs lie in our class.  Fix any edge \(s_k t_k\in E(G_k)\); any
connected spanning $\{s_k,t_k\}$-join becomes a spanning Eulerian
multigraph after adding one copy of \(s_k t_k\), so
\(\OPT(G_k,s_k,t_k)\ge \TSPtour(G_k)-1\ge (9/8-o(1))n_k-1\).  By
\cref{lem:adjacent-lp-exact}, \(\LP(G_k,s_k,t_k)=n_k-1\), so
\[
  \frac{\OPT(G_k,s_k,t_k)}{\LP(G_k,s_k,t_k)}
  \ge \frac{(9/8-o(1))n_k-1}{n_k-1} = \frac98-o(1).
\]
Taking the limsup gives the claim.
\end{proof}

\Cref{thm:gap-intro} follows: the lower bound is \cref{thm:gap-lower}
and the upper bound is \cref{cor:wyy-gap}.

\section{Consequences and Open Problems}\label{sec:discussion}

We proved a fixed-endpoint walk-length theorem on simple $2$-connected
subcubic graphs: the bound $\lfloor(5n+n_2(G))/4\rfloor-1$ holds for
every terminal pair outside the exceptional case of
\cref{thm:endpoint-intro}, where we obtain $\lfloor(5n+n_2(G))/4\rfloor$,
and its coefficient $5/4$ is best possible already when $n_2(G)=0$.
On cubic $3$-edge-connected graphs the same bound confines the path
Held--Karp integrality gap to $[9/8,5/4]$ (\cref{thm:gap-intro}), by a
direct argument that does not use the general path-to-tour reduction.
Whether $\lfloor(5n+n_2(G))/4\rfloor-1$ also holds in the exceptional
case is open; it would follow from $\dfc(K^{33},e^\star)\le-3/2$ for a
suitable split, and the split matters (\cref{app:twelve} gives a
$12$-vertex instance with admissible splits of deficit $-1$ and $-3$).

\paragraph*{Beyond Subcubic.}
A natural open question is whether a similar fixed-endpoint
walk-length theorem holds on simple $2$-connected graphs of
maximum degree $4$; both the WYY excess theorem and our
endpoint-conversion construction rely on the subcubic assumption.
Stronger connectivity may force shorter walks: every $4$-connected
planar graph is Hamiltonian-connected~\cite{Tho83}, so
$\OPT(G,s,t)=n-1$ on that class.

\section*{Disclosure of Interests}
The author has no competing interests to declare that are relevant to
the content of this article.

\appendix

\section{A Twelve-Vertex Instance}\label{app:twelve}

Let $G$ be the graph with vertex set
$\{s,t,z_1,z_2,x,y,a,b,c,a',b',c'\}$ and edge set
\[
  \{sa,\,sb,\,sc,\ \ ax,\,bx,\,xz_1,\ \ z_1z_2,\,z_1y,\ \
    ya',\,yb',\,a't,\,b't,\ \ tc',\,c'z_2,\,cz_2\}.
\]
Then $G$ is simple, $2$-connected, and subcubic with $n=12$ and
$n_2(G)=6$ (the degree-$2$ vertices are $a,b,c,a',b',c'$); the
terminals $s,t$ are non-adjacent with $\deg_G(s)=\deg_G(t)=3$, and
$n+n_2(G)=18\equiv2\pmod4$, so $(G,s,t)$ lies in the exceptional
case of \cref{thm:endpoint-intro}.  Moreover $G-\{s,t\}$ is
connected, so every assignment of terminal incidences satisfies the
component condition of \cref{lem:component-aware}, and every
resulting rooted instance $(K^{33},e^\star)$ is admissible and not a
rooted $\theta$-chain
(\cref{lem:gadget-admissibility,lem:component-aware}).

The $\{3,3\}$ construction leaves free which neighbour of $s$ is
reattached to $s_2$ and which neighbour of $t$ is reattached to
$t_2$; there are nine such splits.  Exhaustive enumeration of the
even covers of $K^{33}$ through $e^\star$ shows that for the split
$(s_2\leftarrow c,\ t_2\leftarrow c')$ the minimum excess is $7$,
giving
\[
  \dfc(K^{33},e^\star)=(7-2)-\tfrac{16+8}{4}=-1,
\]
whereas for each of the other eight splits the minimum excess is
$5$, giving $\dfc(K^{33},e^\star)=-3$.  Thus the refinement
$\dfc\le-1$ of \cref{thm:wyy} is tight on a fixed admissible split,
and no argument that treats all admissible splits alike can certify
$\dfc\le-3/2$; only a suitable choice of split attains $-3$.  For
the split $(s_2\leftarrow c,\ t_2\leftarrow c')$ the instance has
the same shape as the graph $G_1$ of the remark following
\cref{thm:wyy}: the vertices $s_2,t_2,z_1,z_2$ are joined by
$e^\star$, the edge $z_1z_2$, and four internally disjoint subcubic
chains, so contracting each chain to a single edge yields $K_4$.

Finally, $\OPT(G,s,t)=14$: the walk
\[
  s,\,b,\,x,\,a,\,s,\,c,\,z_2,\,z_1,\,y,\,b',\,y,\,a',\,t,\,c',\,t
\]
has length $14$, and exhaustive search over connected spanning
$\{s,t\}$-joins shows that none is shorter.  Since
$14\le15=\lfloor(5n+n_2(G))/4\rfloor-1$, the instance is consistent
with the open question of whether the bound
$\lfloor(5n+n_2(G))/4\rfloor-1$ extends to the exceptional case.

\section{The Scan Estimate}\label{app:scan}

This appendix records why the estimate $\dfce$ of
\cref{lem:wyy-constructive} satisfies $\dfce(G,e)\le-1$ whenever
$G-e$ is $2$-connected and $(G,e)$ is not a rooted $\theta$-chain.
In the notation of~\cite{WYY22}, $\dfce(G,e)$ is the value
$\Delta(G,e)$ computed by \textsc{Scan}~\cite[Alg.~1]{WYY22}, and
the even-cover routine $\mathtt{EC}(G,e,\Delta)$ of~\cite[\S6]{WYY22}
is organised into exactly three branches according to
$\Delta(G,e)\in\{-\tfrac12,-1,-\tfrac32\}$.  The input specification
of the branch $\Delta=-\tfrac12$~\cite[Alg.~4]{WYY22} states that,
for a root edge with $G-e$ simple and $2$-connected, this branch
arises precisely when $(G,e)$ is a rooted $\theta$-chain.  Hence,
when $(G,e)$ is not a rooted $\theta$-chain, the computed value
satisfies $\Delta(G,e)\le-1$, which is the last sentence of
\cref{lem:wyy-constructive}.

\section{Citation Table for \texorpdfstring{\cite{WYY22}}{Wigal--Yoo--Yu}}\label{app:citations}

\Cref{tab:wyy-citations} lists the statements of~\cite{WYY22} used
in this paper and where each enters the argument.

\begin{table}[H]
\centering\footnotesize
\begin{tabular}{@{}ll@{}}
\toprule
Statement in~\cite{WYY22} & Role in this paper \\
\midrule
\S2: even covers, $\exc$, $\mathcal E(G,e)$, deficit
  & conventions of \cref{sec:wyy} \\
\S2, Fig.~1: subcubic chains; rooted $\theta$-chains
  & definitions before \cref{thm:wyy} \\
Thm.~2.4~(T1): deficit $\le-\tfrac12$ with equality case
  & \cref{thm:wyy} \\
Thm.~2.4~(T3): instances of deficit $-1$
  & remark following \cref{thm:wyy} (not used) \\
Alg.~1 (\textsc{Scan}) and its analysis in \S6
  & estimate $\dfce$ in \cref{lem:wyy-constructive} \\
Alg.~2 and the $\mathtt{EC}$ branches of \S6
  & even-cover output; \cref{app:scan} \\
Cor.~6.8: consolidated correctness and $O(n^2)$ time
  & \cref{lem:wyy-constructive} \\
\S7: discussion of the cubic lower bound
  & proof of \cref{thm:tight-intro} \\
\bottomrule
\end{tabular}
\caption{Statements of~\cite{WYY22} used in this paper, ordered by
where they enter the argument.}
\label{tab:wyy-citations}
\end{table}

\section{Implementation Notes for \texorpdfstring{\cref{alg:pathtsp}}{Algorithm 1}}\label{app:impl}

We record the data-structure choices behind the $O(n^2)$ bound of
\cref{cor:algo-wyy}; the number of edges is $O(n)$ throughout
because all graphs involved are subcubic.  The graph is stored as
adjacency lists with edge identifiers, so splitting a terminal,
adding $e^\star$, and subdividing the split edges (Cases~B and~C)
are $O(n)$ pointer operations.  The component-aware split of
\cref{lem:component-aware} needs the components of $G-\{s,t\}$ and,
for each component, one edge to $s$ and one to $t$; a single
depth-first search finds them in $O(n)$ time.  The calls to
\textsc{Scan} and to the even-cover routine of
\cref{lem:wyy-constructive} take $O(n)$ and $O(n^2)$ time,
respectively~\cite[Cor.~6.8]{WYY22}.  Building the
component-contraction multigraph of \cref{lem:walk}, extracting a
spanning tree, and lifting its edges to preimages are again one
depth-first search over $O(n)$ edges, and the resulting join has
$n+\exc(F)-3=O(n)$ edge-copies.  The final contraction of $B_s,B_t$
and the Euler trail (Hierholzer's algorithm) are linear.  Hence
every step except the even-cover routine runs in $O(n)$ time, and
the total is $O(n^2)$.

\end{document}